\documentclass[11pt,a4paper]{article}
\usepackage[margin=26mm]{geometry}

\usepackage{algorithm,algpseudocode}
\usepackage{amsmath,amssymb,amsfonts,amsthm,mathrsfs,url,color}
\usepackage{authblk}
\usepackage{cite}
\usepackage{comment}
\usepackage{enumitem}
\usepackage{fancybox}
\usepackage{fancyhdr}
\usepackage{fbox}
\usepackage{graphicx}
\usepackage{import}
\usepackage{textcomp}
\usepackage{thm-restate}
\usepackage{titling}
\usepackage[nottoc,notlot,notlof]{tocbibind}
\usepackage{xcolor}
\usepackage{xspace}

\usepackage{hyperref,multirow}
\hypersetup{
colorlinks = true,
citecolor = violet,
linkcolor = violet,
}

\declaretheoremstyle[
  headfont=\bfseries,
  bodyfont=\itshape,
  notefont=\bfseries,
  spaceabove=7pt,
  spacebelow=7pt,
]{mdef}
\declaretheorem[style=mdef,name=Research Question]{rquestion}
\declaretheorem[style=mdef,name=Theorem,numberwithin=section]{theorem}
\declaretheorem[style=mdef,name=Conjecture,numberwithin=section,sibling=theorem]{conjecture}
\declaretheorem[style=mdef,name=Example,numberwithin=section,sibling=theorem]{example}
\declaretheorem[style=mdef,name=Definition,numberwithin=section,sibling=theorem]{definition}
\declaretheorem[style=mdef,name=Lemma,numberwithin=section,sibling=theorem]{lemma}
\declaretheorem[style=mdef,name=Proposition,numberwithin=section,sibling=theorem]{proposition}
\declaretheorem[style=mdef,name=Observation,numberwithin=section,sibling=theorem]{observation}

\declaretheoremstyle[
  headfont=\bfseries,
  bodyfont=\itshape,
  notefont=\bfseries,
  spaceabove=0pt,
  spacebelow=0pt,
]{claimsty}

\DeclareMathOperator{\csp}{CSP}
\DeclareMathOperator{\Csp}{CSP}

\DeclareMathOperator{\Aut}{Aut}

\DeclareMathOperator{\proj}{proj}

\DeclareMathOperator{\V}{\mathcal{V}}
\DeclareMathOperator{\constraints}{\mathcal{C}}
\DeclareMathOperator{\instance}{\mathcal{I}}

\newcommand{\sA}{\mathbb A}
\newcommand{\sB}{\mathbb B}

\newcommand{\sD}{\mathbb D}
\newcommand{\sI}{\mathbb I}

\newcommand{\sX}{\mathbb X}

\newcommand{\efpp}{efpp}
\newcommand{\true}{\text{TRUE}}

\newcommand{\unreach}{Graph Unreachability\xspace}
\newcommand{\reach}{Graph Reachability\xspace}
\newcommand{\diunreach}{Digraph Unreachability\xspace}
\newcommand{\direach}{Digraph Reachability\xspace}

\newcommand{\tuple}[1]{\mathbf{#1}}
\newcommand{\tup}[1]{\mathbf{#1}}
\newcommand{\rel}[1]{\mathbb{#1}}

\DeclareMathOperator{\scope}{S}
\DeclareMathOperator{\Var}{\text{Var}}

\providecommand{\keywords}[1]
{
  \small
  \textit{\textit{Keywords}:} #1
}
\renewcommand{\today}{\ifcase \month \or January\or February\or March\or %
  April\or May \or June\or July\or August\or September\or October\or November\or %
December\fi, \number \year}

\title{Constraint Satisfaction Problems over Finitely Bounded Homogeneous Structures: a Dichotomy between FO and L-hard\thanks{Funded by the National Science Centre, Poland under the Weave-UNISONO call in the Weave programme, project no. 2021/03/Y/ST6/00171}} 
\date{\today}
\author{Leonid Dorochko}
\author{Micha{\l} Wrona}
\affil[1]{Faculty of Mathematics and Computer Science, Jagiellonian University, Poland}

\pretitle{\begin{center}\LARGE\bfseries}
\posttitle{\par\end{center}\vskip 0.5em}

\begin{document}

\maketitle
\begin{abstract}
Feder-Vardi conjecture, which proposed that every finite-domain Constraint Satisfaction Problem (CSP) is either in P or it is NP-complete, has been solved independently by Bulatov and Zhuk almost ten years ago. Bodirsky-Pinsker conjecture which states a similar dichotomy for  countably infinite first-order reducts of finitely bounded homogeneous structures is wide open.  

In this paper, we prove that CSPs over first-order  
expansions of finitely bounded homogeneous model-complete cores are either first-order definable (and hence in non-uniform AC$^0$) or L-hard under first-order reduction. 
It is arguably the most general complexity dichotomy when it comes to the scope of structures within Bodirsky-Pinsker conjecture. 
Our strategy is that we first give a new proof of Larose-Tesson theorem, which provides a similar dichotomy over finite structures, and then generalize that new proof to  infinite structures.    
\end{abstract}
\keywords{constraint satisfaction problems, finitely bounded homogeneous structure, infinite structures, first-order definability, L-hardness, computational complexity, dichotomies}

\section{Introduction}
The goal of computational complexity is to describe the amount of certain resources, such as time and space, required to solve a computational problem under consideration. Instead of working on a single problem, it is reasonable to work on the whole family of problems admitting a uniform formulation.
A general formalism that allows the formulation of many natural computational problems is the Constraint Satisfaction Problem $\csp(\sA)$, parametrized by a relational structure $\sA$. Every $\sA$ contributes to a different computational problem, e.g., a Boolean satisfiability problem, graph coloring, or solving equations over a finite field. An instance of $\csp(\sA)$ is a set of constraints built out of variables and relations in $\sA$. The goal is to satisfy all constraints simultaneously by the same assignment to variables. Feder-Vardi conjecture~\cite{FederVardi}, settled in 2017 independently by Bulatov~\cite{Bulatov:2017} and Zhuk~\cite{Zhuk:2017, Zhuk:2020} states that whenever $\sA$ is finite, $\csp(\sA)$ is either in P or it is NP-complete. 
That result stands in opposition to Ladner's theorem~\cite{Ladner} on the existence of NP-intermediate problems on the condition that P$\subsetneq$NP. Although obtaining the dichotomy was the original motivation, the algorithms for all polynomial-time $\csp(\sA)$ are arguably of no less importance. They generalize Gaussian elimination~\cite{IdziakMMVW10} and local-consistency techniques~\cite{BartoKozikBoundedWidth}, which are ubiquitous in constraint solving. Thus, they provide the same way of solving many natural problems in P. On the other hand, a simple algebraic condition provides a polynomial time gadget reduction from monotone 1-in-3-SAT to any NP-complete $\csp(\sA)$~\cite{BJK05}.

Although the holy grail of the CSP community has been achieved, it is not the end of the story.
Indeed, there is a number of active research directions generalizing $\csp(\sA)$ over finite structures $\sA$. In particular, Bodirsky-Pinsker 
 conjecture~\cite{BPP-projective-homomorphisms} extends Feder-Vardi conjecture to some well-behaved infinite structures $\sA$ which cover a number of problems  not expressible as $\csp(\sA)$
if $\sA$ is finite, e.g., in temporal reasoning~\cite{temporalCSP,MakarychevMZ15}, phylogenetic reconstruction~\cite{PhylogenyCSP,ChatziafratisM23} and graph theory~\cite{Schaefer-Graphs,HomogeneousGraphs}. The well-behaved structures $\sA$ are described in two steps. 
We start with a Fra\"{i}ss\'e class~\cite{Hodges} of finite structures $\mathcal{C}$ given by a finite number of finite structures (obstructions) which are forbidden to embed into the members of $\mathcal{C}$. The Fra\"{i}ss\'e
limit of $\mathcal{C}$ is the basic structure $\sB$. By  Fra\"{i}ss\'{e} theorem, the structure $\sB$ is homogeneous\footnote{All notions are carefully defined in Section~\ref{sect:prelim}.} and by the choice of $\mathcal{C}$ it is finitely bounded. Then, one considers and asks about the complexity of $\csp(\sA)$ for structures $\mathbb{A}$ first-order definable in $\mathbb{B}$, known as first-order reducts of $\mathbb{B}$.

\begin{conjecture} \textbf{(Bodirsky-Pinsker)}
\label{conj:BodirskyPinsker}
For every finite-signature first-order reduct $\sA$
of a finitely boun\-ded homogeneous structure $\sB$, $\csp(\sA)$ is in P or it is NP-complete.
\end{conjecture}

Although not resolved in full generality, the conjecture has been confirmed in the number of particular cases such as all first-order reducts of $(\mathbb{Q},<)$, known also as temporal (constraint) languages~\cite{temporalCSP}, first-order reducts of the so-called rooted triple relation, known also as phylogeny languages~\cite{PhylogenyCSP}, or all first-order reducts of all homogeneous graphs~\cite{Schaefer-Graphs,HomogeneousGraphs}. Another classification considers CSPs that are expressible in the logic MMSNP~\cite{MMSNP-journal}. 
Despite these attempts, it is fair to say that the general conjecture is wide open. 

At the same time, the boundary between P and NP-complete has not been the only one studied for finite-domain CSPs~\cite{DualitiesCSP}. In~\cite{LaroseT09TCS}, one can find a number of not yet resolved conjectures on which CSPs are, e.g., in L (logarithmic space) or NL (non-deterministic logarithmic space). The questions have not yet been answered despite certain partial results in the literature; see, e.g.,~\cite{BartoKW12,EgriHLR14,DalmauEHLR15,CarvalhoDK08}.
Furthermore,~\cite{LaroseT09TCS} contains the following dichotomy built on the characterization of first-order definable finite-domain CSPs from~\cite{FO_char1}.

\begin{restatable}{theorem}{larosetesson}
\label{thm:LaroseTesson}
\textbf{(Larose-Tesson)}
Let $\sA$ be a finite structure over a finite signature. Then one of the following holds. 
\begin{itemize}
\item $\csp(\sA)$ is L-hard (under first-order reduction).
\item  $\csp(\sA)$ is first-order definable, and hence  in non-uniform $\text{AC}^0$.
\end{itemize}
\end{restatable}

Recall that a canonical example of an L-hard problem is \reach, where given an undirected graph $G$ and two distinguished vertices $s$ and $t$, we ask whether there is a path from $s$ to $t$~\cite{Reingold08}. More appropriate in the context of constraint satisfaction problems is the \unreach problem, the complement of the \reach problem, which is L-hard but can actually be expressed as $\csp(\{0,1 \}; \{ (0) \}, \{ (1) \}, \{ (0,0), (1,1) \})$. A similar problem, but for directed graphs: \direach is easily NL-hard. By the Immerman-Szelepcsényi theorem~\cite{Immerman88,Szelepcsenyi88}, the \diunreach is also NL-hard. The latter problem is equivalent to   $\csp(\{0,1 \}; \{ (0) \}, \{ (1) \}, \{ (0,0), (1,1), (0,1) \})$.
 \emph{Non-uniform $\text{AC}^0$}, as already mentioned in the abstract, is the complexity class consisting of all languages accepted by polynomial-size constant-depth
families of Boolean circuits with AND/OR gates of unlimited fan-in. The proof of the Larose-Tesson theorem and the proof of our main contribution below show  first-order definability in the easy case, from which follows the containment in non-uniform AC$^0$ (see, e.g., Theorem 6.4 in \cite{Libkin}). 
Recall that there are problems in L which are not in non-uniform $\text{AC}^0$~\cite{ParityCircuits}, i.e., $\text{non-uniform AC}^0 \subsetneq L$. In particular, any problem that is L-hard under first-order reduction cannot lie in non-uniform $\text{AC}^0$.
To see that the dichotomies between L-hard and non-uniform AC$^0$ are non-trivial, observe that PARITY is an example of a computational problem in L that is provably neither in $\text{AC}^0$ nor L-hard; see, e.g.,~\cite {VollmerCircuit} for an introduction to circuit complexity. 

Since the question on the P versus NP-complete dichotomy for first-order reducts of finitely-bounded homogeneous structures is notoriously resistant and the conjectures from~\cite{LaroseT09TCS} on CSPs in L and NL are not even solved for finite structures, it is reasonable to ask if at least Theorem~\ref{thm:LaroseTesson} may be generalized. In this paper, we achieve this goal for a huge class of structures within Conjecture~\ref{conj:BodirskyPinsker}. 

\begin{restatable}{theorem}{main}
\label{thm:main}
\textbf{(Main theorem)}
Let $\sA$ be a finite-signature first-order expansion of a finitely bounded homogeneous model-complete core $\sB$.   Then one of the following holds. 
\begin{itemize}
\item $\csp(\sA)$ is L-hard (under first-order reduction).
\item  $\csp(\sA)$ is first-order definable, and hence in non-uniform $\text{AC}^0$.
\end{itemize}
\end{restatable}

In the algebraic approach to the computational complexity of CSPs, instead of $\csp(\sA)$, one usually considers $\csp(\sA')$, which is equivalent to $\csp(\sA)$ but $\sA'$ is a structure with better properties known as the core of $\sA$, or in the case of infinite-domain CSPs, a model-complete core. Theorem~\ref{thm:LaroseTesson} was proved without loss of generality for finite core structures. For infinite structures, we cannot simply restrict to model-complete cores 
--- it is not even known if a model-complete core of a reduct of a finitely bounded homogeneous structure is a reduct of a finitely bounded homogeneous structure --- see Question~(5) on the list of open problems in~\cite{BodirskyBook}, in Section~14.2. (Note that this question has been settled affirmatively in the case of 
Ramsey structures~\cite{Cores_Ramsey}.) Irrespective of the answer to this question, our proof requires the assumption that $\sA$ and $\sB$ have the same automorphism group, or equivalently, they are first-order interdefinable. Thus, $\sA$ needs to be a first-order expansion of $\sB$. We agree that it restricts the scope of infinite structures under consideration. Regardless, we still believe that our result is essentially the most general complexity classification when it comes to the scope of Conjecture~\ref{conj:BodirskyPinsker}.
Indeed, the partial dichotomies mentioned above usually start with a very simple basic structure such as $(\mathbb{Q}; <)$ or a homogeneous graph, and then consider all its first-order reducts. Despite the development and use of certain general methods~\cite{BodPin-CanonicalFunctions,Smooth-journal}, all these proofs are tailored to specific basic structures and are not generalizable. The only notable exception is the classification of CSPs definable in the logic MMSNP, but that proof relies on representing CSPs as MMSNP formulae.

\begin{example}
\label{ex:thm-main}
The order $(\mathbb{Q}; <)$ over rational numbers is a finitely bounded homogeneous
model-comple\-te core. Theorem~\ref{thm:main} covers all its expansions by first-order definable relations. In this case, all corresponding CSPs are L-hard, in fact, even NL-hard.

On the other hand, consider $(V; E, N)$ which is the expansion of the countable universal homogeneous graph (or the countable homogeneous triangle-free graph) by the \emph{non-edge} relation which holds on different vertices not connected by an edge.  Its CSP is first-order definable (and hence in non-uniform AC$^0$) but, as we will see, some of its first-order expansions give rise to L-hard CSPs.
\end{example}

In order to prove Theorem~\ref{thm:main}, we first provide a new proof of Theorem~\ref{thm:LaroseTesson} and then generalize the new proof to infinite structures covered by Theorem~\ref{thm:main}.

\subsection{Proof strategy}

\paragraph{Finite duality.} Throughout the paper, we freely change the perspective on $\csp(\sA)$, seeing its instances sometimes as sets of constraints over relations in $\sA$ to be solved and sometimes as structures $\mathbb{C}$ over the same signature as $\sA$ which are to be verified whether they homomorphically map to $\sA$ or not. Using the latter perspective, one can show that $\csp(\sA)$ is first-order definable if and only if $\sA$ has finite duality, i.e.,  a finite set of \emph{obstructions} $\mathcal{O}$ such that $\mathbb{C}$ admits a homomorphism to $\sA$ if and only if no $\sD \in \mathcal{O}$ maps homomorphically to $\mathbb{C}$. For the result, see Theorem 5.6.2 in~\cite{BodirskyBook}, proved using Rossman's preservation theorem~\cite{Rossman08}. For finite structures, we first knew this characterization of first-order definable CSPs thanks to Atserias~\cite{Atserias05}.

\paragraph{\efpp-definitions and first-order reductions.} Further, we observe that a finite core
struc\-ture $\sA$ with $|A| \geq 2$ containing the
equality $=$ and all constants in its signature gives rise to an $L$-hard $\csp(\sA)$. A first-order
reduction from \unreach can now be simply achieved by marking a constant $s$ in $\sA$ as
a source, some other constant $t$ as a target, and introducing a constraint $x = y$ between vertices
connected by an edge in the graph. Therefore, as it is sometimes the case, we cannot include the
equality in the signature of $\sA$ if it is not explicitly there, and we cannot use primitive
positive definitions with a free use of equality as a reduction between considered CSPs. Hence,
instead of primitive positive formulae, we use equality-free primitive positive formulae
(\efpp) where only atoms built out of relational symbols in the signature of
$\sA$ are used,~see e.g.~\cite{Borner08} for different types of pp-formulae. We show in Section~\ref{sect:prelim} that \efpp-definitions (provided by
\efpp-formulae) give rise to first-order reductions between the corresponding CSPs.

\paragraph{Proof sketch for Theorem\texorpdfstring{~\ref{thm:LaroseTesson}}{ 1.2}} Without loss of generality, we assume that $\sA$ is a core and contains all constants. An \efpp-formula $\phi(x,y)$ with  free variables $x$ and $y$ is called a $(C, x, D, y)$-imp\-li\-cation in $\sA$ if $C,D$ are \efpp-definable, strictly contained in the projections of $\phi$ to $x$ and $y$, respectively, and if $D = \{ y \mid \exists x \in C \wedge R(x,y)  \}$.
A $(C, x, D, y)$-imp\-li\-cation in $\sA$ is \emph{balanced} if $C = D$ and the projections of $\phi$ to $x$ and $y$ coincide.  
Now, if there is an \efpp-definition of a balanced implication in $\sA$, then we can also \efpp-define an equivalence relation on a subset of the domain $A$ of $\sA$ with at least two equivalence classes, which together with all constants naturally admit a first-order reduction from \unreach. On the other hand, if a balanced implication is not \efpp-definable, then $\csp(\sA)$ is solvable by the $(1)$-minimality algorithm, also known as generalized arc-consistency. We reformulate this well-known algorithm, so that we could produce what we call tree \efpp-formulae defining every domain, also empty, deduced by it on an instance of $\csp(\sA)$. A tree \efpp-formula defining an empty set over an unsatisfiable instance gives rise to an obstruction. We show that the number of indistinguishable obstructions is finite, or from a huge enough tree \efpp-formula witnessing the unsatisfiability of some instance, we may produce a balanced implication.

\paragraph{Proof sketch for Theorem\texorpdfstring{~\ref{thm:main}}{ 1.3}} Our main contribution is proved in a similar fashion. The main differences are that we need to consider implications as formulae over pairs of tuples of variables rather than pairs of variables.

\begin{example}
\label{ex:impl-graphs}
Consider an expansion $\sB = (V,E,N)$ of
a universal countable homogeneous graph by a binary non-edge relation $N$.
This structure is a model-complete core. It is homogeneous, in fact, 2-homogeneous, i.e., every orbit is definable by a conjunction of binary atoms: with $E, N$ and $=$. It is also finitely bounded: a structure $\mathbb{C}$ embeds into $\sB$ if and only if $E^{\mathbb{C}}$ contains neither $(x,x)$ (an $E$-loop) nor $(x,y)$ without $(y,x)$  (an $E$-arrow) and $N^{\mathbb{C}}$ contains neither $(x,x)$ (an $N$-loop) nor $(x,y)$ without $(y,x)$  (an $N$-arrow) and there are no $\{x,y\}$ marked both as an edge and a non-edge or marked neither as en edge nor as a non-edge. Since all the bounds are of size at most $2$, $\sB$ is a $2$-bounded structure that satisfies the prerequisites needed for Theorem~\ref{thm:main}.

 The group of automorphisms of $\sB$ is transitive, that is, every element can be mapped onto any other by an appropriate automorphism. Therefore, one cannot define a non-trivial subset of $V$, and hence an implication over a pair of variables as defined above for a finite structure. On the other hand, $\phi(x_1, x_2, y_1, y_2) \equiv (E(x_1, x_2) \wedge E(y_1, y_2)) \vee (N(x_1, x_2) \wedge N(y_1, y_2))$ is, according to 
 Definition~\ref{def:inf-implication}, 
 an $(E, (x_1, x_2), E, (y_1, y_2))$-implication. It is quite easy to see that there is a first-order reduction from \unreach to $\csp(\sB')$ where $\sB'$ extends $\sB$ with the relation defined by $\phi$. We replace any vertex with a pair of variables, adding $E(x_1, x_2)$ for the source and $N(y_1, y_2)$ for the target. On the other hand, the structure $\sB$ is easily seen to have finite duality (the only obstructions are an $E$-loop, an $ N$-loop, and a two-element structure which imposes both $E$ and $N$ on its elements). Hence, $\csp(\sB)$ is first-order definable and in non-uniform AC$^0$.  
\end{example}

A very similar example may be given
 for the universal homogeneous triangle-free graph.
We consider the rational order below.

\begin{example}
\label{ex:impl-Q}
The structure $(\mathbb{Q}; <)$ is $2$-homogeneous and $3$-bounded.
Observe also that $\phi(x,y,z) \equiv y < z$ is a $(<, (x, y) , <,  (x,z))$-implication. By Proposition~\ref{prop:balimpLhard}, it follows that $\csp(\mathbb{Q}; <)$ is L-hard under first-order reduction.
\end{example}

In general, a finite-signature finitely bounded homo\-ge\-neous model-complete core $\sB$ is $k$-homo\-ge\-neous and $\ell$-bounded for some $k,\ell \in \mathbb{N}$. For its first-order expansion $\sA$, either we may \efpp-define an implication similar to these from examples above or, as we prove, $\csp(\sA)$ may be solved by establishing $(k, \max(k, \ell))$-minimality. The $(k,\max(k,\ell))$-minimality algorithm is then reformulated so that we can have so-called $k$-tree \efpp-formulae of all at most $k$-ary relations derived by the minimality algorithm. Again, either there are finitely many  $k$-tree \efpp-formulae witnessing unsatisfiable instances, and hence a finite set of obstructions, or we can \efpp-define an implication as in the examples above. Although based on similar arguments as the proof for finite structures, the proof for infinite structures appears to be much more involved when it comes to details. 

\paragraph{Implications in other papers.} Although  in other contexts,  
implications themselves have been used before. Over finite structures, e.g. in~\cite{AbsSubalgebras}, one writes $C^{+ R}$ to denote $D = \{ y \mid \exists x \in C~R(x,y) \}$. It is, unfortunately, somehow troublesome to generalize this notation to infinite structures where one has to deal also with non-binary relations or formulae over more than two free variables. It is why we stick to the implications defined above (and also in more detail below). We note that they were also used in research on bounded strict width~\cite{Wrona:2020a,Wrona:2020b,NagyP24} and the tractability of quasi Jonsson operations~\cite{NagyPW-jonsson} over infinite structures.

\subsection{Significance  and the future work}

Our main result, Theorem~\ref{thm:main}, is arguably the most general dichotomy theorem concerning structures within the Bodirsky-Pinsker conjecture. Indeed, other dichotomies concern either
first-order reducts of some simple structures or problems definable in some logics.  

The theorem was obtained by first reproving the analogous result for finite structures and then generalizing the proof to infinite structures. A similar general idea has been adopted in~\cite{BartoBodorKozikMottetPinsker,BrunarKNP25} in order to provide new loop lemmata over $\omega$-categorical structures: 
\begin{flushleft}
\emph{old proofs for finite structures are not easily generalizable to infinite structures, so what about providing new proofs that can be more easily generalized?}
\end{flushleft}

This way of thinking brings a new hope for open problems over infinite-domain CSPs. We may just need a slightly different proof of the Bulatov-Zhuk theorem to generalize it to first-order reducts of finitely bounded homogeneous structures. Such a new proof may be, of course, hard to obtain straightaway, but we believe that the methods from this paper may be used to identify both finite and infinite structures of bounded width (solvable by the minimality algorithm) that give rise to CSPs in NL and L. We discuss this possible direction at the end of the paper in Section~\ref{sect:futurework} when our methods are already known.

\subsection{Organization of the paper}

All necessary notions are defined and explained carefully in Section~\ref{sect:prelim}. The new proof of Theorem~\ref{thm:LaroseTesson} is provided in Section~\ref{sect:finite} and the proof of our main contribution in Section~\ref{sect:infinite}. Both sections display some similarities, but the advantage of providing both proofs separately is that the reader interested only in finite-domain CSPs does not have to go through some model-theoretic notions characteristic of infinite-domain CSPs.

\section{Preliminaries}
\label{sect:prelim}
We write $[n]$ for the set $\{1, \ldots, n \}$. Most of the definitions from below may also be found in~\cite{BodirskyBook}.

\subsection{Structures and formulae}
All structures in this article are assumed to be relational and at most countable.  
Let $\sA$ be a relational structure, we denote the domain of $\sA$ by a plain letter $A$. Sometimes, for the sake of simplicity, we use the same letter $R$ both to denote a relational symbol and the corresponding relation $R^{\sA}$. Let $\phi$
be a first-order formula over the signature of $\sA$ and free variables $\Var(\phi)$. We identify
the interpretation $\phi^{\sA}$ of $\phi$ in $\sA$ with the set of satisfying assignments $a:
\Var(\phi) \to A$ for $\phi$. Let $\tuple u = (u_1, \ldots, u_k)$ be a tuple of elements of
$\Var(\phi)$. We define the \emph{projection of $\phi^{\sA}$ onto $\tuple u$}  as the relation
$\proj_{\tuple u}(\phi^{\sA}):=\{f(\tuple u)\mid f\in \phi^{\sA}\}$ where $f(\tuple u)$ is a
shorthand for $(f(u_1), \ldots, f(u_k))$ used throughout the text. 

A \emph{first-order expansion} of a structure $\sB$ is an expansion of  $\sB$ by relations which are first-order definable in $\sB$, i.e., for which there exists $\phi$ over free variables $\tuple u$  such that  $R = \{ f(\tuple u) \mid f \in \phi^{\sA} \}$.  

We write $\Aut(\sA)$ to denote the set of automorphisms of $\sA$. An orbit of a $k$-tuple $\tuple a = (a_1, \ldots, a_k)$ in $\sA$ where $a_1, \ldots, a_k \in A$ is the set 
$\{ \alpha(\tuple a) \mid \alpha \in \Aut(\sA) \}$. An orbit $O$ is \emph{injective} if the projection of $O$ to any two different coordinates is contained in disequality, and it is \emph{almost injective} if the projection to any two different coordinates is either contained in disequality or consists only of one pair of elements $(a,a)$ with $a \in A$. 
A countable relational structure $\sA$ is \emph{$\omega$-categorical} if its automorphism group is \emph{oligomorphic}, i.e., for all $n \geq 1$, the number of orbits of $n$-tuples in $\sA$ is finite. 
For an $\omega$-categorical relational structure $\sA$ and for every $n\geq 1$ it holds that two $n$-tuples of elements of the structure are in the same orbit in $\sA$ if and only if they have the same \emph{type}, i.e., they satisfy the same first-order formulae. The \emph{atomic type} of an $n$-tuple is the set of all atomic formulae satisfied by this tuple; the structure $\sA$ is \emph{homogeneous} if any two tuples of the same atomic type belong to the same orbit. For $k\geq 1$, we say that a relational structure $\sA$ is \emph{$k$-homogeneous} if for all tuples $\tuple a, \tuple b$ of arbitrary finite equal length, if all respective $k$-subtuples of $\tuple a$ and $\tuple b$ are in the same orbit in $\sA$, then $\tuple a$ and $\tuple b$ are in the same orbit in $\sA$. If a homogeneous $\sA$ has a finite signature, then it is $\omega$-categorical and $k$-homogeneous for some $k \in \mathbb{N}$. 
For $\ell\geq 1$,  we say that $\sA$ is \emph{$\ell$-bounded} if for every finite $\rel X$, if all substructures $\rel Y$ of $\rel X$ of size at most $\ell$ embed in $\sA$, then $\rel X$ embeds in $\rel A$. A structure $\sA$ is \emph{finitely bounded} if it is $\ell$-bounded for some $\ell$. In this case there is a finite set of structures $\mathcal{F}_{\sA}$ over at most $\ell$ elements such that some finite structure $\sX$ embeds into $\sA$ if and only if no structure from 
$\mathcal{F}_{\sA}$ embeds into $\sX$.

A structure $\sA$ is a \emph{model-complete core} if for all endomorphisms of $\sA$ and all finite subsets of its domain, there exists an automorphism of $\sA$ which agrees with the endomorphism on the subset.  If $\sA$ is finite and satisfies this condition,  we simply say that it is a \emph{core}.

A first-order formula is called  \emph{primitive positive} (pp) if it is built exclusively from
atomic formulae, existential quantifiers, and conjunctions. Primitive positive formulae usually
allow equalities even if they are not in the signature of $\sA$. We say that a first-order formula
is an \emph{equality-free primitive positive} (\efpp) formula for a relational structure $\sA$ if it is a pp-formula which allows only atomic formulae built from relational symbols in the signature of $\sA$.
A  relation is \emph{pp-definable} or \emph{\efpp-definable} in a relational structure  if it is
first-order definable by a pp-formula or an \efpp-formula, respectively.

It is well known that for any $\omega$-categorical model-complete core $\sA$, for all $n \geq 1$, all orbits
of $n$-tuples in $\sA$ are pp-definable~\cite{cores}. For \efpp-definitions, we have the following.

\begin{observation}
\label{obs:core}
Let $\sB$ be a model-complete core, then every almost injective orbit of tuples in $\sB$ is \efpp-definable in $\sB$.
\end{observation}

\begin{proof}
    
Let $k \geq 1$ and $O$ be an almost injective orbit of $k$-tuples in $\sB$. By~\cite{cores}, there exists a pp-definition $\phi(x_1, \ldots, x_k)$ of $O$ with free variables $x_1, \ldots, x_k$. We assume without loss of generality that $\phi$ consists of a quantifier prefix and a quantifier-free part.
The goal is to remove all equalities from $\phi$ without affecting the relation defined by $\phi$. 
We first add to $\phi$ all equalities of the form $x = x$ for every variable $x$ in the definition of $\phi$, we also add $x = y$ if it contains $y = x$ or there exists $z$ such that both $x = z$ and $z = y$ are in $\phi$. We add equalities to $\phi$ in this way as long as the fixed point is not reached.
Then, the pairs of variables $x,y$, such that $x = y$ is in the new $\phi$, form an equivalence relation $E$.

We now consider two cases. 
If $k=1$, then there are no equalities of the form $x_i = x_j$ with $i, j \in [k]$ such that $i \neq j$, and hence every equivalence class of $E$ contains at most one element in $\{ x_1, \ldots, x_k \}$.
Thus, we can replace all variables in $\phi$ by the representatives of their equivalence class. If possible, we choose $x_i$ for $i \in [k]$. Now, we can safely remove all equalities from $\phi$ without affecting the relation it defines.
The second case is where $k \geq 2$. Now we can have $x_i=x_j$ in $\phi$ with $i, j \in [n]$ such that $i \neq j$.  
In this case, since $O$ is almost injective, the projection of $\phi$ to $x_i$ and $x_j$ is $C \subseteq A$ with $|C| = 1$. By the previous case, $C$ is \efpp-definable in $\sA$. Thus, we can add a constraint $C(y)$ for every variable $y$ in the equivalence class of $E$ that contains both $x_i$ and $x_j$. Then we can safely remove all equalities between the variables in this equivalence class. In the same vein, we remove equalities in other equivalence classes of $E$ containing at least two variables in $\{x_1, \ldots, x_k\}$. If there is one or zero such variables in the considered equivalence class, then we proceed as in the case where $k = 1$. Hence, in this case as well, we showed how to remove equalities from $\phi$ without changing the relation it defines.
\end{proof}

\subsection{CSPs and  the minimality algorithm}
For the relation $R\subseteq A^k$ over a non-empty set $A$ and $i_1,\ldots,i_{\ell}\in[k]$, we write $\proj_{(i_1,\ldots,i_{\ell})}(R)$ for the $\ell$-ary relation $\{(a_{i_1},\ldots,a_{i_{\ell}})\mid (a_1,\ldots,a_k)\in R\}$.
The number $k$ is the \emph{arity} of $R$.
The \emph{scope}, i.e., the set of all entries of a tuple $\tuple t$
is denoted by $\scope(\tuple t)$.
We write $I_k^A$ for the relation containing all injective, i.e., with pairwise different entries, $k$-tuples of elements in $A$, and $\binom{S}{k}$ to denote the set of all $k$-element subsets of the set $S$. 

Let $R \subseteq A^k$ and $U$ be a set of variables. Then a constraint $C$ over $R$ and with the scope $U$ is a subset of $A^U$ such that there exists an enumeration $u_1,\ldots,u_k$ of the elements of $U$ so that for all $f\colon U\to A$, it holds that 
$f\in C$ if, and only if, $(f(u_1),\dots,f(u_k))\in R$.
An \emph{instance of $\csp(\sA)$} is a pair $\instance=(\V,\constraints)$, where $\V$ is a finite set of variables and $\mathcal C$ is a finite set of \emph{constraints} over relations in $\sA$ and whose scopes are subsets of $\mathcal{V}$. 
The relational structure $\sA$ is called the~\emph{template} of $\csp(\sA)$.
A~\emph{solution} of a CSP instance~$\instance = (\V,\constraints)$ is a mapping $f\colon \V\rightarrow A$ such that for every $C\in\constraints$ with scope $U$, $f|_U\in C$.

An instance $\instance=(\V,\constraints)$ of $\csp(\sA)$ is \emph{non-trivial} if it does not contain any empty constraint; otherwise, it is \emph{trivial}.
Given a constraint $C\subseteq A^U$ with $U\subseteq \V$ and a tuple $\tuple{v}\in U^k$ for some $k\geq 1$, the \emph{projection of $C$ onto $\tuple{v}$} is defined by $\proj_{\tuple{v}}(C):=\{f|_{\scope(\tuple{v})} \mid f \in C\}$.

\begin{definition}\label{def:minimality}
Let $1\leq k\leq \ell$. We say that an instance $\instance=(\V,\constraints)$ of $\csp(\sA)$ is \emph{$(k,\ell)$-minimal} if both of the following hold:
\begin{enumerate}
\item \label{algmin:full} the scope of every tuple of elements of $\V$ of length at most $\ell$ is contained in the scope of some constraint in $\constraints$;
\item \label{algmin:proj} for every $m\leq k$, for every tuple $\tuple u\in \V^m$, and for all constraints $C_1, C_2 \in \constraints$ whose scopes contain the scope of $\tuple u$, the projections of $C_1$ and $C_2$ onto $\tuple u$ coincide.
\end{enumerate}
We say that an instance $\instance$ is \emph{$k$-minimal} if it is $(k,k)$-minimal.
\end{definition}

Let $1\leq k$. If $\instance=(\V,\constraints)$ is a $k$-minimal instance and $\tuple u$ is a tuple of variables of length at most $k$, then there exists a constraint in $\constraints$ whose scope contains the scope of $\tuple u$. 
Then the projection of this and all other constraints in $\instance$ containing $\tuple u$ to $\tuple u$ is the very same constraint, it is denoted by $\proj_{\tuple u}(\instance)$ and called the \emph{projection of $\instance$ onto $\tuple u$}.

Let $1\leq k\leq \ell$, let $\sA$ be an $\omega$-categorical structure.
Clearly not every instance $\instance=(\V,\mathcal{C})$ of $\csp(\sA)$ is $(k,\ell)$-minimal.
However, it is well known that every instance $\instance$ is \emph{equivalent} to a $(k,\ell)$-minimal instance $\instance'$ of $\csp(\sA')$ where $\sA'$ is the expansion of $\sA$ by all at most $\ell$-ary 
pp-definable relations in the sense that $\instance$ and $\instance'$ have the same  solution set. 
In particular, if $\instance'$ is trivial, then $\instance$ has no solutions. Moreover, $\instance'$ can be computed from $\instance$ in polynomial time by enforcing both conditions in the natural way. First, by introducing a full constraint over an $l$-ary subset of $\mathcal{V}$ when necessary, and then by repeatedly removing assignments from constraints so that the second condition was satisfied. In the remainder of this paper, we refer to this algorithm as the $(k,l)$-minimality algorithm. 

For the purposes of this paper, we need a slightly stronger result. Indeed, we will show that it is enough to look at the expansion $\sA'$ of $\sA$ by all \efpp-definable relations, i.e.,  every instance $\instance$ is \emph{equivalent} to a $(k,\ell)$-minimal instance $\instance'$ of $\csp(\sA')$ where $\sA'$ is the expansion of $\sA$ by all at most $\ell$-ary 
\efpp-definable relations. 

\begin{observation}
\label{obs:efppmin}
Let $\instance$ be an instance of $\csp(\sA)$ and $\sA'$ an expansion of $\sA$ that contains all \efpp-definable relations in $\sA$ of arity at most $\ell$. Then there exists a $(k,\ell)$-minimal instance $\instance'$ of $\csp(\sA')$ equivalent to $\instance$ where $\sA'$ is an expansion of $\sA$ by all at most $\ell$-ary \efpp-definable relations. 
\end{observation}

\begin{proof}
Indeed, we first add a full constraint over an $l$-ary subset of $\V$ whenever necessary. All full relations are  
\efpp-definable in $\sA$ by empty formulae. In order to enforce the second condition, we repeatedly look at some 
$C_1, C_2 \in \constraints$ and a set of variables $U$ of size at most $k$ contained in both the scope of $C_1$ and $C_2$ and replace $C_1$ with some $C'_1$ so that the projection of $C'_1$ to $U$ was contained in the projection of $C_2$ to $U$. That is, we replace $C_1$ with 
$C_1' := \{ f \in C_1 \mid \exists g \in C_2~g|_{U} = f|_{U} \}$. Observe that $C'_1$ is a constraint over the relation with an \efpp-definiton: $\exists v_1 \ldots v_t~R_1(\tuple{x_1}) \wedge R_2(\tuple{x_2})$ where $U = \scope(\tuple{x_1}) \cap \scope(\tuple{x_2})$ and $\{ v_1, \ldots, v_t\} = \scope(\tuple{x_2}) \setminus U$ and $C_i$ with $i \in [2]$ is a constraint over $R_i$.
\end{proof}

 We say that the $(k,l)$-minimality algorithm solves $\csp(\sA)$ if it transforms a given instance $\instance$ into a trivial instance if and only if $\instance$ does not have a solution. In this case, one says that $\sA$ has \emph{relational width} $(k,l)$. If $\sA$ has relational width $(k,k)$, then it has relational width $k$.

Alternatively, an instance $\instance$ of $\csp(\sA)$ may be seen as a finite structure over the same signature as $\sA$. One can transform a set of constraints $\constraints$ over variables $\V$  into an appropriate structure just by taking $\V$ as its domain and including a tuple $\tuple u$ of variables into $R^{\instance}$ always when there is a constraint $C \subseteq A^U$ with $\scope(\tuple u) = U$ such that $R^{\sA} = \{ f(\tuple u) \mid f \in C \}$. On the other hand a structure $\sI$ may be transformed into the appropriate set of constraints by taking its domain as a set of variables $\V$ and including a constraint $C \subseteq A^{U}$ over the relation $R^{\sA}$ 
always when 
$R^{\sI}$ contains $\tuple u$. 

\begin{definition}
Let $\sA$ be a $\tau$-structure. A set of finite $\tau$-structures $\mathcal{O}$ is called an obstruction set for $\sA$ if, for
any finite $\tau$-structure $\mathbb{C}$, we have that $\mathbb{C}$ maps homomorphically to $\mathbb{A}$ if and only if  no $\sD \in \mathcal{O}$ maps homomorphically to $\mathbb{C}$.

We say that a structure has finite duality if it has a finite obstruction set. 
\end{definition}

\subsection{FO-reductions and \efpp-formulae}

\begin{definition}
\label{def:FO-inter}
Let $\sigma$ and $\tau = \{R_1,\ldots,R_s\}$ be two relational signatures. A $k$-ary first-order interpretation with $p$ parameters of
$\tau$ in $\sigma$ is an $(s+1)$-tuple $\iota =(\phi_U, \phi_{R_1},\ldots,\phi_{R_s})$ of first-order formulae over $\sigma$, where $\phi_U = \phi_U(\tuple x, \tuple y)$ has
$k+p$ free variables $\tuple x = (x_1,\ldots,x_k)$ and 
$\tuple y= (y_1,\ldots,y_p)$ and $\phi_{R_i}(\tuple x^1, \ldots, \tuple x^r, \tuple y)$ has $k \cdot r+p$ free variables where $r$ is
the arity of $R_i$ and each $\tuple x^j =(x_1^j, \ldots ,x_k^j)$ and $\tuple y$ is as above.

Let $\mathbb{C}$  be a $\sigma$-structure. A tuple $\tuple c = (c_1,\ldots,c_p)$ of elements of $\mathbb{C}$ is said to be proper if $c_i\neq c_j$ when $i \neq j$. Let
$\tuple c = (c_1,\ldots,c_p)$ be proper. The interpretation of $\mathbb{C}$  through $\iota$ with parameters $\tuple c$, denoted by $\iota(\mathbb{C}, \tuple c)$, is the $\tau$-structure whose universe is $\{ \tuple a \in C^k : \phi_{U} (\tuple a, \tuple c)\}$
and whose interpretation for $R_i$ is
$$\{(\tuple a^1,\ldots, \tuple a^r) \in (C^k)^r \mid \phi_U(\tuple a^1,\tuple c) \wedge \cdots \wedge \phi_U(\tuple a^r, \tuple c) \wedge \phi_{R_i}(\tuple a^1,\ldots,\tuple a^r, \tuple c)\}$$
\end{definition}

\begin{definition} 
\label{def:FO-red}
Let $\sigma$ and $\tau$ be finite relational signatures, let $\mathcal{C}$ be a class of $\sigma$-structures and let $\mathcal{D}$ be a class of $\tau$-structures
closed under isomorphisms. We say that a first-order 
interpretation $\iota$ with $p$ parameters of $\tau$ in $\sigma$
is a \emph{first-order reduction} of $\mathcal{C}$ to $\mathcal{D}$ if for every $\sigma$- structure $\mathbb{C}$ with at least $p$ elements the following two
equivalences hold:
\begin{enumerate}
\item $\mathbb{C} \in \mathcal{C}$ iff $\iota(\mathbb{C}, \tuple c)  \in \mathcal{D}$ for every proper $\tuple c$,
\item $\mathbb{C} \in \mathcal{C}$ iff $\iota(\mathbb{C}, \tuple c)  \in \mathcal{D}$ for some proper $\tuple c$.
\end{enumerate}
\end{definition}

The next step is to prove a version of Lemma~2.6 in~\cite{LaroseT09TCS} trimmed to \efpp-formulae. The difference between the lemmata is that in our version we disallow equalites.

\begin{lemma}
  \label{lem:efpp-fo-seq}
  Let $\sA$ be an arbitrary relational structure, let $R$ be a relation over $A$ and let $\sA'$ be
  an expansion of $\sA$ by $R$.
  Then the following two conditions are equivalent:
  \begin{itemize}
    \item $R$ has an \efpp-definition in $\sA$,
    \item there exists a finite sequence $\sA = \sA_{0}, \ldots, \sA_{s} = \sA'$, such that each
      relational structure $\sA_{i}$ is obtained from the previous one by one of the following
      operations:
      \begin{enumerate}[label=(\arabic*)]
        \item \label{seqop:removerel} removing a relation,
        \item \label{seqop:rearrange} adding an $r$-ary relation $R'(x_1, \ldots,x_r) \equiv R''(x_{\pi(1)},\ldots, x_{\pi(l)})$ obtained from an $l$-ary $R''$ already in the structure  by rearranging the coordinates in $R''$ according to a 
        map $\pi:[l] \rightarrow [r]$, 
        \item \label{seqop:intersect} adding the intersection of two relations of the same arity,
        \item \label{seqop:product} adding the product of two relations,
        \item \label{seqop:projection} adding a relation obtained by projecting an $r$-ary relation to its first $r-1$
          variables.
      \end{enumerate}
  \end{itemize}
\end{lemma}

\begin{proof}
  We first observe that if there is a sequence from $\sA$ to $\sA'$ as in the formulation of the lemma, then $\sA'$ is \efpp-definable in $\sA$. Indeed, it is clear regarding the first kind of operation. When it comes to the rearranging of the coordinates, the \efpp-definition of $R'$ from $R''$ is given already in the formulation of the lemma. The intersection of two relations $R_1$ and $R_2$ of arity $r$ is \efpp-definable by $R_1(x_1, \ldots, x_r) \wedge R_2(x_1, \ldots, x_r)$, the product of an $r$-ary $R_1$ and an $l$-ary $R_2$ is $R_1(x_1,\ldots, x_r) \wedge R_2(x_{r+1}, \ldots, x_{r+l})$. Finally, the \efpp-definition of the projection of an $r$-ary relation $R'$ to its first $r-1$ coordinates is $\exists x_r~R'(x_1, \ldots, x_r)$.  

  On the other hand, every \efpp-definition $\phi$ is built inductively from simpler \efpp-formulae $\phi_1$ and $\phi_2$ by establishing a conjunction $\phi(x_1,\ldots, x_k, y_1, \ldots, y_l) \equiv \phi_1(x_1, \ldots, x_k) \wedge \phi_2(y_1,\ldots, y_l)$ with possibly $\{ x_1, \ldots, x_k \} \cap \{ y_1, \ldots, y_l\} \neq \emptyset$,  by attaching an existential quantifier in front of a simpler formula, without loss of generality we consider $\phi \equiv \exists y~\phi_1(x_1, \ldots, x_n)$, or by rearranging the variables in $\phi_1$. For the
  latter operation, we simply use the operation from Item~\ref{seqop:rearrange}. When it comes to adding an existential quantifier in front of a formula, we assume that $R_{\phi_1} = \{ (x_1, \ldots, x_k) \mid \phi(x_1, \ldots, x_k) \}$. In order to obtain $R_{\phi}$ 
  \efpp-defined by $\exists y~R_{\phi_1}(x_1, \ldots, x_n)$ we first rearrange the variables so that $y$ was the last variable in the atom, then project the corresponding relation to its first $(r-1)$-coordinates and finally we rearrange the coordinates back so that they were in same order as before with an exception of the missing coordinate corresponding to $y$.
  Finally, we look at the conjunction of two \efpp-formulae. By arranging the variables, we may assume without loss of generality that $R_{\phi}(x_1, \ldots, x_{p+r+l}) \equiv R_{\phi_1}(x_1, \ldots, x_{p+r}) \wedge R_{\phi_2}(x_p, \ldots, x_{p+r+l})$.
  In order to obtain $R_{\phi}$ by the available operations, we first 
  obtain $R'_{\phi_1}$, which is the product of $R_{\phi_1}$ and the full $l$-ary relation as well as $R'_{\phi_2}$ which is the product of the full $p$-ary relation and $R_{\phi_2}$. 
  Observe that a full $m$-ary relation may be obtained by the operation in Item~\ref{seqop:rearrange} so that every variable 
  in $R'$ is a dummy variable that is they do not occur on the right hand side in $R''$. Then $R_{\phi}$ is simply the intersection of $R'_{\phi_1}$ and $R'_{\phi_2}$ which are of the same arity. 
\end{proof}

In~\cite{LaroseT09TCS}, one can find proofs that operations \ref{seqop:removerel} and
\ref{seqop:intersect}--\ref{seqop:projection} give rise to  first-order reductions. See Lemma~2.7 and Lemmas~2.9--2.11, respectively. For the operation in Item~\ref{seqop:rearrange}, we provide the following lemma.

\begin{lemma}
  \label{lem:rearrangefo}
  Let $\sA$ be a relational structure, and $\sA'$ an expansion of $\sA$ by an $r$-ary relation $R'(x_1, \ldots, x_r)\allowbreak \equiv R(x_{\pi(1)}, \ldots, x_{\pi(l)})$ that rearranges coordinates of some $R$ in $\sA$ according to a  map $\pi: [l] \rightarrow [r]$. Then there is a $1$-ary $0$-parameter reduction
  from $\csp(\sA')$ to $\csp(\sA)$.
\end{lemma}
\begin{proof}
 Assume witout loss of generality that $\sA = (A; R_1, \ldots, R_n, R)$.
  We define our $1$-ary $0$-parameter interpretation $\iota = (\phi_{U}, \phi_{R}, \phi_{R_{1}},
  \ldots, \phi_{R_{n}})$ as follows:
  \begin{equation*}
    \phi_{U}(x) = \true
  \end{equation*}
  \begin{equation*}
    \phi_{R_i}(\tup x) = R_{i}(\tup x)
  \end{equation*}
  \begin{equation*}
    \phi_{R}(x_{1}, \ldots, x_{l}) = R(x_{1}, \ldots, x_{l}) \vee \exists z_1 \cdots \exists z_r~\bigwedge_{i = 1}^r R'(z_1, \ldots, z_{r}) \wedge \bigwedge_{i \in [r]: \pi^{-1}(i) \neq \emptyset} \bigwedge_{j \in \pi^{-1}(i)} z_i = x_j  
  \end{equation*}
\end{proof}

\noindent
We are now ready to turn to the main result of this subsection. 

\begin{lemma}
\label{lem:efpp-FO}
Let $\sA$ be a relational structure and $\sA'$ be an expansion of $\sA$ by a relation $R$ which has an \efpp-definition in $\sA$. Then there exists a first-order reduction   from $\Csp(\sA')$ to $\csp(\sA)$ where $\sA'$ is an expansion of $\sA$ by $R$. 
\end{lemma}

\begin{proof}
It follows from Lemma~\ref{lem:efpp-fo-seq} that the relation $R$ may be obtained from $\sA$ by a series of operations allowed by the formulations of the lemma. By Lemma~2.7 and Lemmas~2.9--2.11 in~\cite{LaroseT09TCS} as well as Lemma~\ref{lem:rearrangefo}, and since first-order reductions are preserved by compositions, see e.g.~Proposition 1.28 in~\cite{immerman}, it follows that every sequence of these operations gives rise to a first-order reduction from $\csp(\sA')$ to $\csp(\sA)$. The lemma follows.
\end{proof}

Let $\text{arity}(\sA)$ be the maximum of all arities of relations in $\sA$.
In particular, without affecting the complexity of the problem, we may say that a finite $\sA$ contains
all relations of arity at most $\text{arity}(\sA)$ \efpp-definable in $\sA$ and that an expansion $\sA$ of a $k$-homogeneous $\ell$-bounded structure contains all at most $\max(k,\ell, \text{arity}(\sA))$-ary \efpp-definable relations in $\sA$. 
On the additional assumption that $\sA$ is a model-complete core,
by Observation~\ref{obs:core}, we may  assume that $\sA$ contains all almost injective orbits of  at most $\max(k,\ell, \text{arity}(\sA))$-ary tuples in $\sA$.
All these assumptions will be in effect throughout the paper.
\label{assum:a-cont-efpp}

\section{Reproving Larose-Tesson theorem}

\label{sect:finite}
We now give a formal definition of an implication in a finite structure $\sA$.

\begin{definition}[Implication]
\label{def:finite-implication}
Let $\sA$ be a finite relational structure and $C,  D \subseteq A$  be \efpp-definable in $\sA$ and non-empty. We say that an \efpp-formula $\phi$ over the signature of $\sA$ with distinguished different free variables $\{ u, v \} \subseteq \Var(\phi)$ is a \emph{$(C, u,D, v)$-implication in $\sA$} if all of the following hold:
\begin{enumerate}
    \item $C\subsetneq\proj_{(u)}(\phi^{\sA})$,
    \item $D\subsetneq\proj_{(v)}(\phi^{\sA})$,
    \item for all $f\in \phi^{\sA}$, it holds that $f(u)\in C$ implies $f(v)\in D$,
    \item for every $a \in D$, there exists $f\in\phi^{\sA}$ such that $f(u)\in C$ and $f(v)=a$.
\end{enumerate}
We say that an implication $\phi$ is \emph{balanced} if $C = D$ and
$\proj_{\tuple u}(\phi^{\sA}) = \proj_{\tuple v}(\phi^{\sA})$
\end{definition}

\noindent
We say that $\sA$ is \emph{unbalanced} if there is no balanced implication in $\sA$. The next step is to define a composition of two implications.

\begin{definition}[Composition]
    \label{def:composition-fin}
    Let $\sA$ be a relational structure, let $C, D, E \subseteq A$ be non-empty, let $\phi_1$ be a
    $(C,u^1,D,v^1)$-implication in $\sA$, and  $\phi_2$  a $(D,u^2,E,v^2)$-implication in
    $\sA$ such that $\proj_{(v^1)} \phi_1^{\sA} = \proj_{(u^2)} \phi_2^{\sA}$.
    Let $\phi_2'$ be a copy of $\phi_2$ with variable $u^2$ renamed to $v^1$ and all other variables renamed so that they are not in $\phi_1$. %being exactly the same as in $\phi_2$. 
    We define $\phi_1\circ\phi_2$ to be a $(C,u^1, E,v^2)$-implication  given by the \efpp-formula  $(\phi_1\wedge \phi_2')$.

Let $\psi$ be a balanced $(C,u,C,v)$-implication. For $n\geq 2$, we write $\psi^{\circ n}$ for the \efpp-formula $\psi\circ \cdots \circ \psi$ where $\psi$ appears exactly $n$ times.
\end{definition}

\noindent
We now compose a balanced implication with itself in order to show L-hardness.

\begin{proposition}
\label{prop:balanced-finite-hard}
Let $\sA$ be a finite relational structure that is a core. If there exists a balanced implication in $\sA$, then $\csp(\sA)$ is $L$-hard under FO-reduction.
\end{proposition}

\begin{proof}
Let $\phi$ be a balanced $(C,x,C,y)$-implication in $\sA$ with $\proj_x \phi^{\sA} = \proj_y \phi^{\sA} = D$. Define $\mathcal{B}_{\phi}$ to be a directed graph over elements in $D$ such that there is an arc $a \to b$ in $\mathcal{B}_{\phi}$ if $\phi^{\sA}$ contains a mapping $m$ with $m(x) = a$ and $m(y) = b$.  
A \emph{component} in $\mathcal{B}_{\phi}$ is a strongly connected component that contains at least one arc. A component $F$ in  $\mathcal{B}_{\phi}$ is a \emph{sink} in $\mathcal{B}_{\phi}$ if every arc originating in $F$ ends up in $F$. On the other hand, $F$ is a \emph{source} in $\mathcal{B}_{\phi}$ if every arc targeting an element in $F$ starts also in F. Observe that $\mathcal{B}_{\phi}$ restricted to the vertices in $C$ contains a sink and restricted to $D \setminus C$ a source.

Let $M$ be the least common multiple of the lengths of all cycles in 
$\mathcal{B}_{\phi}$.
Then consider $\eta := \phi^{\circ M}$. Observe that $\mathcal{B}_{\eta}$ contains a loop $(a,a)$ always when $a$ is in a component in $\mathcal{B}_{\phi}$.
Since any two elements connected in $\mathcal{B}_{\eta}$ are also connected in $\mathcal{B}_{\phi}$,
every component in $\mathcal{B}_{\phi}$ is a union of components in 
$\mathcal{B}_{\eta}$, and hence also $\mathcal{B}_{\eta}$ 
contains a loop for all elements existing in some component of 
$\mathcal{B}_{\eta}$.

The next step is to consider $\psi:=\eta^{|D|}$. Observe that 
$\mathcal{B}_{\psi}$ and $\mathcal{B}_{\eta}$ have the same components, and every component in the former graph contains all possible arcs.
Indeed, there is obviously a path of length at most $|D|$ between any two elements in any component in  $\mathcal{B}_{\eta}$, but since every component contains all loops, there is actually a path of length exactly $|D|$ between any two elements in every component. The claim follows.

Now, $\mathcal{B}_{\theta}$ where $\theta(x,y) := \psi(x,y) \wedge \psi(y,x)$ is a disjoint union of components in $\mathcal{B}_{\psi}$. It follows that $R = \{ m(x,y) \mid m \in \theta^{\sA} \}$ is an equivalence relation on a subset of $D$ which contains at least two equivalence classes.

By Lemma~\ref{lem:efpp-FO}, it is now enough to show that there is a first-order reduction from $\csp(\{ 0,1 \}; R_0, R_1,\allowbreak =)$, where $=$ is the equality on $\{0,1 \}$, to $\csp(A; R, R_c, R_d)$, where $R_c = \{ (c) \}$,  $R_d = \{ (d) \}$ and $c,d$ are elements from two different equivalence classes of $R$. 
To this end we will provide a  $1$-ary interpretation with $0$ parameters $\iota$ of the signature $\tau = \{ R, R_c, R_d \}$ in the signature $\sigma = \{ =, R_0, R_1 \}$. It is simple, the formula $\phi_U(x) \equiv \text{TRUE}$, $\phi_{R_c}(x) \equiv R_0(x)$, $\phi_{R_d}(x) \equiv R_1(x)$ and $\phi_{R}(x,y) \equiv \phi_{=}(x,y)$. It is now straightforward to show that an instance $\instance$ of $\csp(\{ 0,1 \}; R_0, R_1, =)$ has a solution if and only if $\iota(\instance)$ is a solvable instance of  $\csp(A; R, R_c, R_d)$.
\end{proof}

\subsection{Rewriting the \texorpdfstring{$(1)$}{(1)}-minimality algorithm}

In this section, we show that if a finite core structure $\sA$ is unbalanced, then it has finite duality. We start with proving that  $\sA$ has relational width $1$.

\begin{proposition} 
\label{prop:noimp1min}
Let $\sA$ be a finite relational structure that is a core.
If $\sA$ is unbalanced, then $\csp(\sA)$ is  solvable by the $(1)$-minimality algorithm. 
\end{proposition}

\begin{proof}
Let $\instance = (\V, \constraints)$ be a non-trivial $(1)$-minimal instance of $\csp(\sA)$.
Define the instance graph $\mathcal{G}_{\instance}$ of $\instance$ to be a digraph whose vertices are pairs of the form $(C,x)$ where 
$x$ is a variable in $\instance$ and $C \subsetneq \proj_{(x)}\instance$ is \efpp-definable in $\sA$. 
There is an arc $(C,x) \rightarrow (D,y)$ in $\mathcal{G}_{\instance}$ if there is a constraint $C \subseteq A^U$ in $\constraints$ such that $R_C(\tuple u)$, where $\tuple u$ satisfies $\scope(\tuple u) = U$
and $R_C = \{ f(\tuple u) \mid f \in C \}$, 
is a $(C, x, D, y)$-implication. 

Since all constants are in $\sA$, we have that either $\proj_x(\instance)$ is a singleton for every variable $x$ in $\V$
or $\mathcal{G}_{\instance}$ is non-empty. In the former case,   since $\instance$ is $(1)$-minimal and non-trivial, the singletons 
$\proj_x(\instance)$ determine the solution of $\instance$ and we are done. In the latter case, we will show that we can replace $\instance$ with an instance $\instance'$ of $\csp(\sA)$ which is non-trivial, $(1)$-minimal,
satisfies $\proj_x \instance' \subseteq \proj_x \instance$ for all $x \in \V$ and $\proj_x \instance' \subsetneq \proj_x \instance$ for at least one $x \in \V$. Observe that it is enough to complete the proof of the proposition. Indeed, we may now repeatedly reduce 
projections of an instance to variables until every projection holds only one element. 

The next step is to observe that $\mathcal{G}_{\instance}$ is acyclic. Suppose it is not. Then, from the constraints on a cycle in
$\mathcal{G}_{\instance}$, we can extract some $(C_i, x_i, C_{i+1}, x_{i+1})$-implications $\psi_i$ in $\sA$ with $i \in [n-1]$ for some $n \in \mathbb{N}$ as well as a $(C_n, x_n, C_1, x_1)$-implication $\psi_n$ in $\sA$ such that $\proj_{(x_1)}(\psi_1) = \proj_{(x_1)}(\psi_n)$. It follows that $\phi := \psi_1 \circ \cdots \psi_{n-1} \circ \psi'_n$, where $\psi'_n$ is a copy of $\psi_n$ with $x_1$ renamed to $x_{n+1}$, is a balanced $(C_1, x_1, C_1, x_{n+1})$-implication. It contradicts the assumption that $\sA$ is unbalanced. 

Let now $(B,x)$ be a sink in $\mathcal{G}_{\instance}$. We set $\instance' = (\V, \constraints')$ to be an instance of $\csp(\sA)$ obtained from $\instance$ by replacing every constraint $C \subseteq A^U$ whose scope contains $x$ with $C' = \{ f \in C \mid f(x) \in B \}$. The goal is to show that $\instance'$ is a non-trivial and $(1)$-minimal instance of $\csp(\sA)$ with $\proj_x \instance' = B$ and $\proj_y \instance' = \proj_y \instance$ for any $y$ different from $x$. The instance $\instance'$ is non-trivial since $\instance$ is $(1)$-minimal. It is an instance of $\csp(\sA)$ because of our assumption that a finite template contains all \efpp-definable relations of arity $\text{arity}(\sA)$. To see that $\instance'$ is also $(1)$-minimal, assume on the contrary that there is a constraint $C' \in \constraints'$, originating from a constraint $C \in \constraints$, with the scope $U$ such that $\proj_y C'$ is not as desired. Clearly $y$ is not $x$, and clearly $x$ is in the scope of $C'$. It follows that $\proj_y C' \subsetneq \proj_{y} C$. Set $D := \{ f(y) \mid f \in C' \}$ and observe that $R_C(\tuple u)$, where $\tuple u$ satisfies $\scope(\tuple u) = U$
and $R_C = \{ f(\tuple u) \mid f \in C \}$, is a $(B, x, D, y)$-implication. It contradicts the fact that $(B,x)$ is a sink in $\mathcal{G}_{\instance}$.
\end{proof}

Let $\sA$ be a finite structure and $\instance = (\V, \constraints)$ be an instance of $\csp(\sA)$. Algorithm~\ref{alg1:cap}  produces a set of constraints
$\mathcal{D}^{\V}(\instance) = \{ D_z \mid z \in \V  \}$. We prove in Proposition~\ref{prop:minAlg1} that a 
$(1)$-minimal instance $\mathcal{I}'$ equivalent to $\instance$ satisfies $\proj_z \instance' = D_z$ for every $z \in \V$.

\begin{algorithm}
\caption{\ \ Let $\sA$ be a finite structure and $\instance = (\V, \constraints)$ an instance of $\csp(\sA)$.}\label{alg1:cap}
\begin{algorithmic}[1]
\ForAll {$u \in \V$} \Comment{Step 1: Initialization}
\label{alg1list:step1}
    \State $D_u \gets \{f(u) \mid f \in C\}$ for some $C \in \constraints$ with $u \in \scope(C)$
    \EndFor

\\

\While{true} \Comment{Step 2: Pruning}
\label{alg1list:step2}
    \State For some $C\subseteq A^U$ in $\constraints$ and
    $u \in \V$ set:
    \State $\widehat{D_{u}} \gets \{ f \in A^{\{u\}} \mid \exists g \in C: (f = g|_{\{u\}}) \wedge \bigwedge_{z \in U} (g|_{\{z\}} \in D_{z}) \}$
    \If {for all $C, u$, their $\widehat{D_{u}}$ contains $D_u$}
        \State \textbf{break}
    \Else
        \State $D_u \gets D_u \cap \widehat{D_{u}}$
    \EndIf
\EndWhile

\\

\State \textbf{return} $\mathcal{D}^{\V}(\instance)$
\Comment{Step 3: Return}
\end{algorithmic}
\end{algorithm}

\begin{proposition}
\label{prop:minAlg1}
Let $\sA$ be a finite  structure, and
$\mathcal{I} = (\V, \constraints)$ be an instance of $\csp(\sA)$.
Then $\mathcal{I}' = (\V, \constraints')$ obtained from $\mathcal{I}$ by replacing every $C \in \constraints$ of the scope $U$ with $$C' = \{ f \in C \mid \bigwedge_{z \in U} f|_{\{ z \}}  \in  D_z \}$$ is $(1)$-minimal with 
$\proj_{(u)} (\instance') = D_u$ for all $u \in \V$
and equivalent to $\mathcal{I}$.
\end{proposition}

\begin{proof}
To see that $\mathcal{I}'$ is equivalent to $\mathcal{I}$, it is enough to observe that whenever Algorithm~\ref{alg1:cap} reduces the set $D_{u}$ for some  $u \in V$, the removed values cannot be a part of a solution to $\instance$.

To prove $(1)$-minimality of $\instance'$ assume on the contrary that there is a constraint $C'$ as defined in the formulation of the lemma and $u \in U$ such that $\proj_{(u)}(C')$ does not contain $D_u$. Then, 
$$\widehat{D_{u}} = \{ f \in A^{\{u\}} \mid \exists g \in C : (f = g|_{\{u\}}) \wedge \bigwedge_{z \in U} (g|_{\{z\}} \in D_{z}) \}$$ does not contain $D_{u}$, which implies that Algorithm~\ref{alg1:cap} could not return $D_u$. It completes the proof of the proposition.
\end{proof}

The next step is to introduce \emph{tree \efpp-formulae}, which play an important role in producing obstructions for unbalanced $\sA$. 

Let $\mathcal{I} = (\V, \constraints)$ be an instance of $\csp(\sA)$ and $\mathcal{X}$ be a set of copies of variables from $\V$ such that every variable in $\V$ has infinitely many copies in $\mathcal{X}$.  
We say that a constraint $C'$ is a copy of a constraint $C \in \constraints$
if it is obtained from $C$ by replacing every variable from $\scope(C)$ by its copy in $\mathcal{X}$. We will also say that an atom $R(\tuple x)$ is a copy of
a constraint $C \in \constraints$ if there exists a constraint $C'$ with the scope $\scope(\tuple{x})$ which is a copy of $C$ over the relation $R$.

\begin{definition}
Let $\mathcal{I} = (\V, \constraints)$ be an instance of $\csp(\sA)$ and  $\mathcal{X}$ a set containing infinitely many copies of every variable in $\V$.
A \emph{tree \efpp-formula $\phi$ over  $\mathcal{I}$} with a distinguished variable  $r_0 \in \Var(\phi)$ called a \emph{root}, is a conjunction of copies of constraints in $\mathcal{C}$ and is in one of the following forms. 

\begin{enumerate}
\item \label{tree-form:1} It is a single atomic formula $R(\tuple{x})$ with $r \in \scope(\tuple x)$. 
\item \label{tree-form:2} It is a conjunction of tree \efpp-formulae $\phi_1, \ldots, \phi_m$ 
with the roots   $r_1, \ldots, r_m$ and a new copy $\phi_0 := R(\tuple x)$  of a constraint in $\constraints$ such that 
\begin{enumerate}
\item $\bigcup_{i \in \{0, \ldots, m \}} \{ r_i \} \subseteq \scope(\tuple x)$, 
\item $\Var(\phi_0) \cap \Var(\phi_i) = \{ r_i \}$ for all $i \in [m]$, and 
\item $\Var(\phi_i) \cap \Var(\phi_j) = \{ r_i \} \cap \{r_j \}$ for all $i,j \in [m]$.
\end{enumerate}
\end{enumerate}

The subformulae $\phi_1, \ldots, \phi_m$ of $\phi$ in Item~\ref{tree-form:2} will be called the \emph{leaf subformulae} of $\phi$ and $\phi_0$ the \emph{root subformula} of $\phi$.
A \emph{formula path} in $\phi$ is a sequence  $\psi_1, \ldots, \psi_{t+1}$ of subformulae of $\phi$ such that $\psi_{t+1}$ is $\phi$, $\psi_1$ is an atomic formula and  for every $i \in [t]$,
$\psi_{i}$ is a leaf subformula of $\psi_{i+1}$.

A tree \efpp-formula is a tree \efpp-formula over $\instance$ for some instance $\instance$ of $\csp(\sA)$. 
\label{def:tree-form}
\end{definition}

Observe that tree \efpp-formulae are quantifier-free.
In the following, we will be interested in $\proj_{r} \phi^{\sA}$ where $\phi$ is a tree \efpp-formula and $r$ is the root of $\phi$. 

\begin{definition}
\label{def:root-defn-finite}
We say that a tree \efpp-formula $\phi$ with the root $r$ is a \emph{minimal tree \efpp-formula over $\instance$} if $\proj_{r} \phi^{\sA} = B$ and $\phi$ is minimal with respect to the number of atoms among all tree \efpp-formulae $\psi$ over $\instance$   with the same root $r$ satisfying 
$\proj_{r} \psi^A = B$.
\end{definition}

\noindent
Note that in the above definition, we allow $B$ to be equal to $\emptyset$. We will now prove that whenever Algorithm~\ref{alg1:cap} derives some new $D_u$ over an instance $\instance$, then there is a tree \efpp-formula $\phi$ over $\instance$ with the root $r$ which is a copy of $u$  so that the projection of $\phi^{\sA}$ onto $r$ is $\{ f(u) \mid f \in D_u \}$ with  $D_u$ from $\mathcal{D}^{\V}(\instance)$.

\begin{proposition}
\label{prop:Alg1-TreeDef}
Let $\sA$ be a finite structure and $\instance$  an instance of $\csp(\sA)$. It follows that whenever Algorithm~\ref{alg1:cap} computes $D_{u}$ either in Step~1 or Step~2, then there exists a tree \efpp-formula $\phi$ over $\instance$ with the root $r$ being a copy of $u$ in $\mathcal{X}$ so that  $\proj_r (\phi)^{\sA} = \{ f(u) \mid f \in D_u \}$.
\end{proposition}

\begin{proof}
The proof goes by the induction of the number of steps executed by Algorithm~\ref{alg1:cap}. For any $D_u$ computed in Step~\ref{alg1list:step1} we just take an atom $R(\tuple x)$ which is a copy of the   constraint $C$ and set the root to a copy of $u$ in $\tuple x$.

Assume now that Algorithm~\ref{alg1:cap} computes $D_{u}$ in Step~2 from $D_{z_1}, \ldots, D_{z_m}$ with 
$U= \{ z_1, \ldots, z_m \}$.
By the induction hypothesis, there are $\phi_1, \ldots, \phi_{m}$ with the roots 
$r_1, \ldots, r_m$, respectively, being the copies  of $z_1, \ldots, z_m$
such that the projection of $\phi_1^{\sA}, \ldots, \phi_{m}^{\sA}$ onto $r_1, \ldots, r_m$ are the sets $\{ f(z_1) \mid f \in D_{z_1} \}, \ldots,$
$ \{ f(z_m) \mid f \in D_{z_m} \}$.
Let $\phi_0 = R(\tuple x)$ be a copy of the constraint $C \in \constraints$ of the scope $U$ used in Step~2 such that 
$\bigcup_{i \in [m]} \{ r_i \} = \scope(\tuple x)$.
Since $\mathcal{X}$ is infinite, we may secure the conditions (2a)--(2c) in Definition~\ref{def:tree-form}.
Observe now that $\psi:=\bigwedge_{i \in [m]} \phi_i$ with the root $r_j$, where $j \in [m]$ and $r_j$ is a copy of $u$,  
 is a tree \efpp-formula over $\instance$ satisfying $\proj_{r_j} \psi^{\sA} = \{ f(u) \mid f \in D_{u} \}$.
\end{proof}

\noindent
We are now ready to prove that unbalanced finite structures have finite duality. 

\begin{proposition}
\label{prop:unbalanced-finite-duality}
Let $\sA$ be a finite core structure over a finite signature. If $\sA$ is unbalanced, then it has finite duality.
\end{proposition}
\begin{proof}
By  Propositions~\ref{prop:noimp1min}, \ref{prop:minAlg1}, and~\ref{prop:Alg1-TreeDef}, for every unsatisfiable instance $\instance$ of $\csp(\sA)$ 
there is 
a minimal tree \efpp-formula $\phi_{\instance}$ over $\instance$ such that the projection of $\phi_{\instance}^{\sA}$ onto its root is $\emptyset$. Without loss of generality, we assume that no proper subformula of $\phi_{\instance}$ that is a tree \efpp-formula has this property.

Let $\Phi'_{\sA}$ be the set of all formulae $\phi_{\instance}$ over all unsatisfiable 
instances $\instance$ of $\csp(\sA)$ and $\Phi_{\sA}$ be a maximal subset of $\Phi'_{\sA}$ such that no two formulae in $\Phi_{\sA}$ may be obtained from each other by renaming the variables.

Assume first that every formula in $\Phi_{\sA}$ has the longest formula path shorter than $d = (2^{\left| A \right|})^2 \cdot |A| + 2$. We will show that in this case $\Phi_{\sA}$ is of finite size. By the induction of $l \leq d$, we will show that there are only finitely many minimal tree \efpp-formulae of the longest formula path shorter than or equal to $l$ up to the renaming of variables. Since the signature of $\sA$ is finite, there are clearly only finitely many different atoms up to the renaming of the variables. All other tree \efpp-formulae are constructed according to Item~\ref{tree-form:2} and hence their longest formula paths are of length strictly greater than $1$. We may therefore assume that the induction hypothesis holds for $l$ and we will prove that it also holds for $(l+1)$.
Consider a tree \efpp-formula composed out of an atom $\phi_0 := R(\tuple x)$ and tree
\efpp-formulae $\phi_1, \ldots, \phi_m$ with the longest formula paths shorter than or equal to $l$.
Since the number of different values in $\sA$ is $|A|$ and  $\phi$ is supposed to be minimal, we
observe that $m \leq |A| \cdot |\scope(x)|$. By the induction hypothesis, there are only finitely many possible $\phi_1, \ldots, 
\phi_m$ up to the renaming of the variables. Since the signature of $\sA$ is finite, we have that there are only finitely many tree \efpp-formulae of the longest formula path shorter than or equal to $l+1$, up to the renaming of the variables. It proves the induction step and completes the proof of the claim.  It follows that $\Phi_{\sA}$ is finite.

We write $\sD_{\phi}$ for the canonical structure  of a tree \efpp-formula $\phi$, that is, a structure whose domain is the set of variables in $\phi$  
and $R^{\sD_{\phi}}$ has a tuple $\tuple{x}$ if and only, if $\phi$ contains an atom $R(\tuple {x})$.

In order to complete the proof of the first case, we will show that 
$\mathcal{O} = \{ \sD_{\phi} \mid \phi \in \Phi_{\sA} \}$ is an obstruction set for $\sA$. Assume first that $\instance$ is unsatisfiable. Then, there exists a minimal  tree \efpp-formula $\phi_{\instance} \in \Phi'_{\sA}$ over $\instance$ such that the projection of
$\phi_{\instance}^{\sA}$ onto its root is $\emptyset$. The formula $\psi_{\instance}$ which is $\phi_{\instance}$ with the variables renamed is therefore in $\Phi_{\sA}$. The canonical structure $\sD_{\psi}$ maps to $\instance$, seen as the structure over domain $\V$, first by renaming the elements of its domain so that they were named as the variables of $\phi_{\instance}$ and then by taking a homomorphism sending every variable in $\phi_{\instance}$, which is a copy of a variable in $\instance$, to the original variable. It follows that if $\instance$ is unsatisfiable, then there is $\sD \in \mathcal{O}$ that maps homomorphically to $\instance$.
On the other hand, the projection of $\phi^{\sA}$ onto its root is $\emptyset$ and hence $\sD_{\phi}$ does not map to $\sA$. Thus, if $\sD_{\phi}$ maps to $\instance$, then it has no solution. It completes the proof of the fact that $\mathcal{O}$ defined above is a finite obstruction set for $\sA$.

To complete the proof of the proposition, we will show that whenever $\Phi_{\sA}$ contains a tree \efpp-formula of a longest formula path longer than $d = (2^{\left| A \right|})^2 \cdot |A| +2$, then there is a balanced implication in $\sA$. Suppose that $\phi'$ is such a tree \efpp-formula and that the projection of $\phi'^{\sA}$ onto its root $r$ is  $\emptyset$. By the assumption made above, the projection of $\phi^{\sA}$ onto its root is a non-empty set for every leaf subformula $\phi$ of $\phi'$. Take such a leaf subformula $\phi$ of $\phi'$ with a longest formula path $\pi$ consisting of $(d-1)$ tree \efpp-formulae: $\psi_1, \ldots, \psi_{d-1}$ with roots $r_1, \ldots, r_{d-1}$, respectively, such that $\psi_1$ is an atomic formula and $\psi_{d-1}$ is $\phi$. Let $\xi$ be $\phi$ with $\psi_1$ removed from it.
Since $\phi$ is minimal, it follows that  the projection $C_{d-1}$ of $\phi^{\sA}$ onto $r_{d-1}$ is strictly contained in 
the projection  $C'_{d-1}$  of $\xi^{\sA}$ onto $r_{d-1}$. In fact, it follows that the projection $C_{i}$ of $\phi^{\sA}$ onto $r_i$ is strictly contained in the projection  $C'_i$ of $\xi^{\sA}$ onto $r_i$ for every $i \in [d-1]$. The roots $r_1, \ldots, r_{d-1}$ do not have to be pairwise different but by the construction of a tree \efpp-formula presented in Definition~\ref{def:tree-form}, we have that if $r_i$ is the same variable as $r_j$ for some $i,j \in [d-1]$, then  $r_j$ is $r_k$ for all $k \in [i,j]$. Since $\phi$ is minimal,
the interval $[i,j]$ cannot contain more than $|A|$ numbers. By the counting argument, there are different roots $r_i, r_j$ with $i < j$ in $[d-1]$ such that $(C_i, C'_{i})$ equals $(C_j, C'_j)$.
Observe that  $\xi$  is a balanced $(C_i, r_i, C_i, r_j)$-implication.
It completes the proof of the proposition.
\end{proof}

\noindent
We are now ready to reprove the Larose-Tesson theorem using the results provided in this section.

\larosetesson*
\begin{proof}
Without loss of generality, we may assume that $\sA$ is a core. If there is a balanced implication in $\sA$, then by
Proposition~\ref{prop:balanced-finite-hard}, we have that $\csp(\sA)$ is L-hard. Otherwise, by~\ref{prop:unbalanced-finite-duality} the structure $\sA$ has finite duality. Hence, $\csp(\sA)$ is first-order definable and it is in non-uniform AC$^0$.
\end{proof}

\section{Infinite CSPs}

\subsection{Hardness results}

\section{Summary and Future Work}
\label{sect:futurework}
The main contribution of this paper is the dichotomy between non-uniform AC$^{0}$ and L-hard for all finite-signature first-order expansions of finitely bounded homogeneous structures --- Theorem~\ref{thm:main}. In order to obtain the result, we first provided a brand new proof of Theorem~\ref{thm:LaroseTesson}, and then generalized this new proof to infinite structures.
We believe that the methods provided may be used further, in particular, to approach the following research question.

\begin{rquestion}
\label{rquestion}
Is it true that every first-order reduct (expansion) of a $k$-homogeneous $\ell$-bounded structure (model-complete core) of relational width $(k,\ell)$ gives rise to CSPs that are in NL or P-complete?
\end{rquestion}

This question is interesting in its own right --- a number of infinite-domain CSPs are solvable by local consistency methods~\cite{BodirskyDalmau}, but it is also directly connected to an open question over finite-domain CSPs. 
Indeed,  it is still not known which finite-domain CSPs are P-complete, and which are in NL. In particular, we do not even know it for problems solvable by the $1$-minimality algorithm. On one hand, in this class of problems, we have the P-complete Horn-3SAT, which is $\csp(\{0,1\}; R_0, R_1, R_{\text{Horn-3SAT}})$ with  $R_0 = \{ (0)\}$, $R_1 = \{ (1)\}$ and
$R_{\text{Horn-3SAT}} = \{ (x,y,z) \in \{0,1 \}^3 \mid (x \wedge y \implies z)\}$, and on the other hand, for instance, structures with finite duality studied in this paper, which clearly give rise to CSPs in NL.  We will now try to convince the reader that the approach in this paper is plausible to generalize to solve this problem and its infinite counterpart. Indeed, we already start with structures solvable by the $(1)$-minimality or the $(k,l)$-minimality algorithm in the case of infinite structures, and hence we would not need to prove the counterparts of Proposition~\ref{prop:noimp1min} or Proposition~\ref{prop:widthunbalanced}, respectively. But we would definitely need different hardness criteria, this time, P-hardness criteria. Perhaps, an implication which resembles a formula of the form $A(x) \implies A(y)$ could have been replaced by a relation more similar to $(A(x) \wedge A(y) \implies A(z))$ that looks like Horn-SAT --- a canonical P-complete problem. %, but over non-Boolean domains. 
One can imagine (we do not know how at the moment) that such a relation may be obtained from a tree or a $k$-tree \efpp-formula that looks like a large enough binary tree, similarly, as to how we obtained balanced implications from a high enough tree or a $k$-tree \efpp-formulae in the proofs of Proposition~\ref{prop:unbalanced-finite-duality} and Proposition~\ref {prop:noimpfindual}, respectively. If all tree or $k$-tree formulae in $\sA$ originating from unsatisfiable instances omit large binary trees, then likely they have bounded pathwidth duality, which implies that the corresponding CSPs are in NL~\cite{DalmauPathWidth}. What we just presented is definitely not a solution to the Research Question~\ref{rquestion}, but it is not impossible that, in this way or another, the proofs in this paper may influence a solution to this important problem.

\bibliographystyle{alpha}
\bibliography{references}

\end{document}